\documentclass[onecolumn,12pt,draft]{IEEEtran}
\usepackage{amsmath}
\usepackage[final]{graphicx}
\usepackage{amsmath,epsfig,bm}
\usepackage{graphicx}
\usepackage{amssymb}
\usepackage{mathrsfs}
\usepackage{comment}
\usepackage{cite}

\usepackage[dvips,dvipdfm]{geometry}
\usepackage{setspace}
\makeatletter
\newcommand{\figcaption}{\def\@captype{figure}\caption}
\makeatother

\newtheorem{thm}{Theorem}

\newcommand{\pr}{\ensuremath{\text{Pr}}}

\newcommand{\dr}{d}
\newcommand{\du}{\partial}

\newcommand{\berx}{\ensuremath{{\text{P}_{\rm e}(\rho x)}}}
\newcommand{\avber}{\ensuremath{\overline{\text{P}}_{\rm e}(\rho, N)}}

\newcommand{\eber}{\ensuremath{\text{E}_{\mathcal{N}} \left[\overline{\text{P}}_{\rm e}(\rho, \mathcal{N}) \right]}}
\newcommand{\beren}{\ensuremath{\overline{\text{P}}_{\rm e}(\rho, \lambda)}}

\newcommand{\cdfrn}{\ensuremath{F_{\gamma_{\rm s}}(x)}}

\newcommand{\avcap}{\ensuremath{\overline{C}(\rho, N)}}

\newcommand{\ecap}{\ensuremath{\text{E}_{\mathcal{N}} \left[\overline{C}(\rho, \mathcal{N})\right]}}

\newcommand{\outprob}{\ensuremath{\text{P}_{\rm{out}}(\rho,N,R)}}
\newcommand{\aveoutprob}{\ensuremath{\text{P}_{\rm{out}}(\rho,\mathcal{N},R)}}
\newcommand{\meanoutprob}{\ensuremath{\text{P}_{\rm{out}}(\rho,\lambda,R)}}

\newcommand{\capen}{\ensuremath{\overline{C}(\rho,\lambda)}}

\newcommand{\ld}{\ensuremath{\lambda}}

\newcommand{\al}{\ensuremath{\alpha}}

\newcommand{\cdfinv}{\ensuremath{F_{\gamma_{\rm s}}^{-1}(e^{-u})}}
\newcommand{\eu}{\ensuremath{e^{-u}}}

\begin{document}
\title{Cognitive Radio with Random Number of Secondary Users}
\author{Ruochen Zeng, Cihan Tepedelenlio\u{g}lu, \emph{Member, IEEE}
\footnote{\scriptsize{Ruochen Zeng and C. Tepedelenlio\u{g}lu are with the Ira Fulton School of Engineering, Arizona State University, Tempe, AZ 85287, USA. (Email:
zengrc@asu.edu, cihan@asu.edu,).}}}
\date{\today}
\maketitle
\vspace{-0.6in}

\begin{abstract}
A single primary user cognitive radio system with multi-user diversity at the secondary users is considered where there is an interference constraint between secondary and primary users. The secondary user with the highest instantaneous SNR is selected for communication from a set of active users which also satisfies the interference constraint. The active number of secondary users is shown to be binomial, negative binomial, or Poisson-binomial distributed depending on various modes of operation. Outage probability in the slow fading scenario is also studied. This is then followed by a derivation of the scaling law of the ergodic capacity and BER averaged across the fading, and user distribution for a large mean number of users. The ergodic capacity and average BER under the binomial user distribution is shown to outperform the negative binomial case with the same mean number of users. Moreover, the Poisson distribution is used to approximate the user distribution under the non-i.i.d interference scenario, and compared with binomial and negative binomial distributions in a stochastic ordering sense. Monte-Carlo simulations are used to supplement our analytical results and compare the performances under different user distributions.
\end{abstract}

\vspace{-0.1in}


\section{Introduction} \label{sec: intro}
Conventional wireless communication systems face the challenge of scarcity of available spectrum resources, and cognitive radio is considered as an ideal architecture to address this problem \cite{788210}. Most cognitive radio paradigms can be categorized into two kinds: overlay and underlay. The overlay paradigm relies on efficient and accurate sensing algorithms to detect the idleness of the primary users (PUs) so that secondary users (SUs) only transmit during these idle times \cite{788210, 6054064, 4493828, 6178840, 5434164}. In the underlay paradigm, which is the focus of this paper, SUs transmit simultaneously with PUs, where interference power at the primary receiver is kept below a certain threshold to satisfy an interference constraint \cite{4840529}. Capacity for Gaussian multiple-input multiple-output (MIMO) channels under received-power constraints in underlay cognitive radio systems is studied in \cite{4069138}.

Multi-user diversity (MUD) has been considered in this context for opportunistic communications of cognitive SUs \cite{4786488}. A widely adopted assumption is that SUs' transmit powers are adjusted to satisfy a peak interference constraint at the primary receiver. Subject to this constraint, the SU with the highest instantaneous SNR is selected for communication. Under this assumption, statistics of the SU transmit SNR in the high power region is studied in \cite{4786488}. When taking into account the interference introduced by the PU at the secondary receiver, MUD gain of the signal-to-noise-plus-interference ratio (SINR) under cognitive multiple-access channel (MAC), broadcast channel (BC), and parallel-access (PAC) are investigated in \cite{5403611}. The CDF expressions of the SINR under MAC, BC \cite{6133630}, and PAC \cite{5577781} are derived to analyze the BER performance. Another common assumption in cognitive MUD systems requires that SUs satisfy an average transmit and interference power constraint at the primary receiver \cite{4100173}. In this scheme, secondary link capacity is shown to scale like $O(M\log\log N)$ as a function of the number of SUs $N$ and available primary spectra $M$ \cite{5454289,5054705}.

In most existing cognitive radio MUD systems, all secondary transmitters scale down their transmit power to meet the interference constraint if the instantaneous peak transmit power causes too much interference. After this potentially continuous power adjustment, the user with the best instantaneous SNR at the secondary receiver is chosen \cite{4786488, 5403611, 5577781}. This scheme requires accurate continuous feedback of the interference channel. We consider an uplink underlay cognitive radio system setup with a single PU and multiple SUs, each equipped with a single antenna. All secondary transmissions obey a pre-determined interference constraint at the primary receiver. The secondary receiver, which is the base station (BS), dynamically updates an index set which contains a list of SUs that satisfy the interference constraint, which creates a random number of SUs. This can be realized with the presence of a single-bit feedback channel between primary receiver and BS \cite{6279522} to inform users whether they are active or passive.  

For the first time in the literature, we consider the effect of having a random number of active users on the performance analysis of cognitive radio system with MUD. In this paper, we study the asymptotic behavior of ergodic capacity and BER averaged across the fading and the user distribution with large mean number of SUs. We also derive non-asymptotic closed form expressions for average BER under several user distributions. Then we consider the non-homogeneous interference case. Furthermore, a stochastic ordering approach is adopted to compare the system performances under different active user distributions. 

The rest of the paper is organized as follows. Section \ref{sec: System_model} and \ref{sec: math_pre} present the system model and some useful mathematical preliminaries. Section \ref{sec: outprob} derives the outage probability under different user distributions. Section \ref{sec: erg_cap} investigates properties of ergodic capacity under different user distributions. Section \ref{BER} derives the closed form expression of average BER under binomial and negative binomial (NB) user distributions, respectively. Section \ref{sec: poisapprox} studies the non-i.i.d interference channels, in which number of users follows a sum of Bernoulli variables with different parameters, which we term Poisson-binomial (PB) distribution, following \cite{la_cam}.  Section \ref{sec: user_order} discusses stochastic ordering of different user distributions. Section \ref{sec: simulations} presents numerical simulations of ergodic capacity and average BER to corroborate our analytical results. Section \ref{sec: conclusions} concludes our work.

\section{System Model}\label{sec: System_model}
We consider an uplink cognitive radio system with multiple SUs, a single PU, and one base station (BS) which serves as the receiver to the SUs. Both the BS and users are assumed to have a single antenna. 
\begin{center}
\includegraphics[height=9cm,width=12cm]{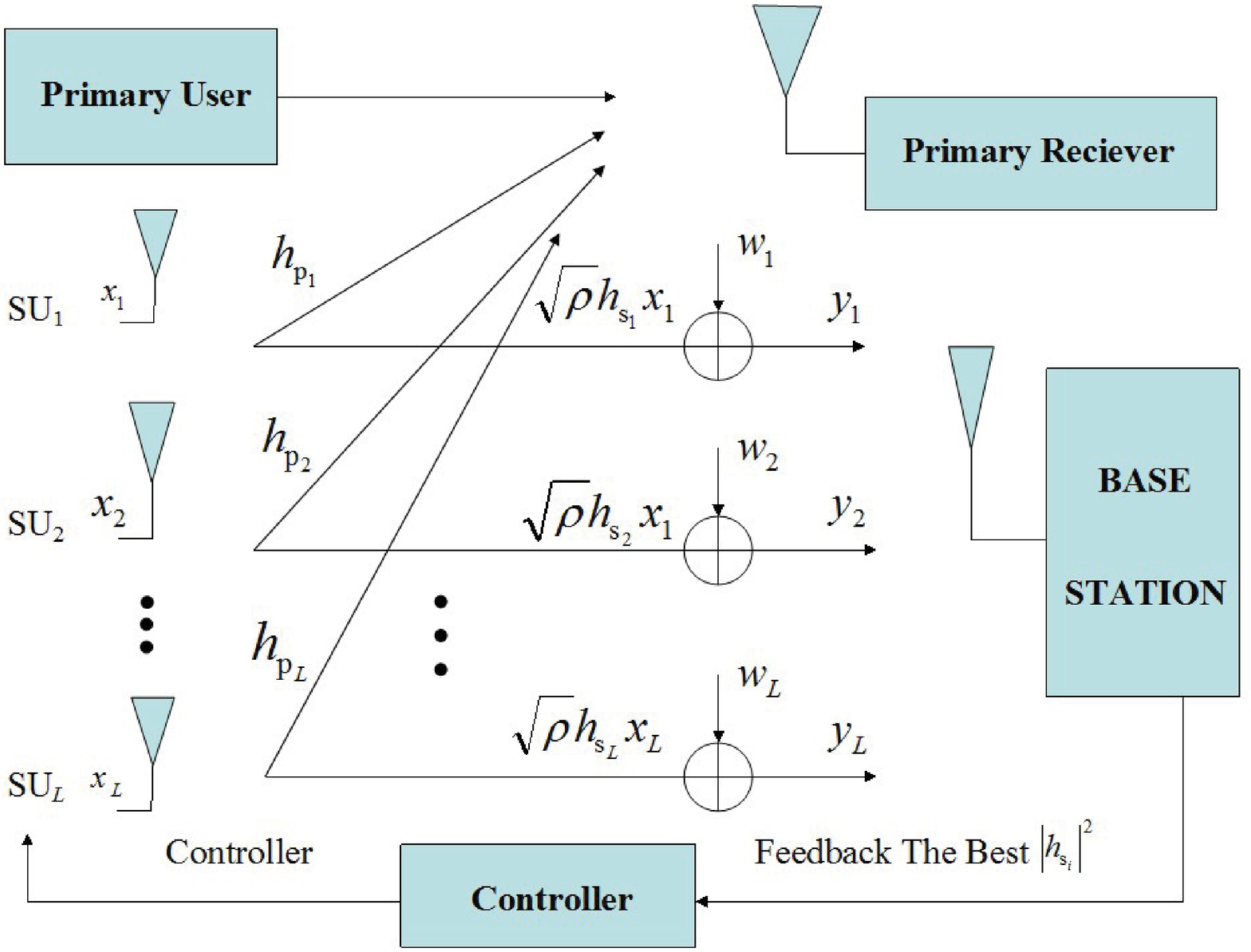}
\end{center}
\vspace{-0.5in}
\figcaption{System Model of Cognitive Radio System}
\label{fig: system_model}
\vspace{0.2in}
As shown in Figure \ref{fig: system_model}, we consider a cognitive radio system with a total of $L$ SUs where the MUD scheme is applied to the secondary system. A SU is allowed to share the spectrum with a primary link as long as the interference power to the primary receiver is less than a threshold $Q$. The received signal from the $i^{th}$ SU at the BS can be expressed as,
\begin{equation}\label{eqn: rcvd_sig_cog}
y_i = \sqrt{\rho}h_{{\rm s}_{i}} x_i+w_i, \hspace{0.4 in} i = 1, 2,\ldots, L,
\end{equation}
where $h_{{\rm s}_{i}}$ denotes the channel coefficient from the $i^{th}$ SU to the BS, $x_i$ is the transmitted symbol, $w_i$ is white Gaussian noise (AWGN). The average received power $\rho$ at the BS is assumed to be identical across SUs. The channel gain of the $i^{th}$ SU at the secondary BS can be expressed as $\gamma_{{\rm s}_{i}} = |h_{{\rm s}_{i}}|^2$, whereas the interference channel gain of the $i^{th}$ SU at the primary receiver is $\gamma_{{\rm p}_{i}} = |h_{{\rm p}_{i}}|^2$. The channel gain of the selected user is denoted by
\begin{align}\label{best_gain}
{\gamma_{\rm s}^{\ast}} = \max_{\{i|i\in \textit{S}\}}\{|h_{{\rm s}_{i}}|^2\},
\end{align}
where $\textit{S}$ is a subset of the users that respect the interference constraint. Consequently, SUs either transmit with fixed power $\rho$, or remain silent, so that a simple transmitter with a fixed power level and one bit feedback is sufficient. In contrast, previous work \cite{4786488, 5403611, 5577781} assumes that secondary transmit power is adjusted to $Q/\gamma_{{\rm p}_{i}}$ if interference constraint is violated, which requires feedback of instantaneous CSI of the interference channel and a sophisticated transmitter to support infinite power levels.  

The distribution of the cardinality of $\textit{S}$ will be specified when different SU distributions are studied. Since all SUs have i.i.d. fading channels to the secondary BS, the subscript $i$ will be dropped when deriving the cumulative distribution function of $\gamma_{{\rm s}_{i}}$. Let $\mathcal{N}$ be the cardinality of $\textit{S}$. Conditioned on $\mathcal{N}=k$, the CDF of the channel gain of the chosen user can be obtained using elementary order statistics as $F_{\gamma_{\rm s}}^k(x)$. To obtain the CDF of ${\gamma_{\rm s}^{\ast}}$ in \eqref{best_gain} we have

\begin{equation} \label{eqn: cdf_gen}
F_{{\gamma_{\rm s}^{\ast}}}(x)= {\text{E}}_{\mathcal{N}}\left[
F_{{\gamma_{\rm s}}}^{\mathcal{N}}(x) \right]= \sum_{k = 0}^{\infty}
{\pr}\left[\mathcal{N}= k \right] F_{\gamma_{\rm s}}^k(x)=
U_{\mathcal{N}}(F_{\gamma_{\rm s}}(x))
\end{equation}
where $U_{\mathcal{N}}(z) = \sum_{k=0}^{\infty} {\pr}\left[\mathcal{N} = k \right]z^k$, $0\leq z\leq 1$, is the probability generating function (PGF) of $\mathcal{N}$.

\section{Mathematical Preliminaries}\label{sec: math_pre}
In this section, we introduce some mathematical preliminaries that will be useful throughout the paper.

\subsection{Completely Monotonic Functions}
A non-negative function $\tau(x):\mathbb{R}^+ \rightarrow \mathbb{R}$ is {\it completely monotonic} (c.m.) if its derivatives alternate in sign \cite{shaked_stochastic_1994}, i.e.,
\begin{equation}\label{eqn: math_prelim_one}
(-1)^k \frac{d^k \tau(x)}{\dr x^k} \geq 0, \hspace{0.2in} \forall x,
\hspace{0.2in} k = 1, 2, 3, \ldots.
\end{equation}
We are also interested in positive functions
whose first-order derivatives are c.m., which are said to have a completely monotonic derivative (c.m.d.). Due to a well-known theorem by Bernstein \cite[pp. 22]{shaked_stochastic_1994}, an equivalent definition for c.m. function is that it can be expressed as a positive mixture of decaying exponentials:
\begin{equation}\label{eqn: math_prelim_two}
\tau(x) = \int_0^{\infty} e^{-sx} d\psi(s)
\end{equation}
for some non-decreasing function $\psi(s)$.

\subsection{Laplace Transform Ordering}\label{sec: LT_ordering}
In this section we introduce {\emph{Laplace transform}} (LT) ordering, a kind of stochastic ordering, to compare different user distributions. This stochastic ordering will be useful in comparing error rate and ergodic capacity averaged across user and channel distributions. LT order is a partial ordering on non-negative random variables \cite[pp. 233]{stochastic_ordering}.

Let $\mathcal{X}$ and $\mathcal{Y}$ be non-negative random variables. $\mathcal{X}$ is said to be less than $\mathcal{Y}$ in the LT order (written $\mathcal{X} \leq_{\rm{Lt}} \mathcal{Y}$), if
 $\text{E}[e^{-s\mathcal{X}}] \geq \text{E}[ e^{-s \mathcal{Y}}]$ for all $s>0$. 
An important theorem found in \cite{shaked_stochastic_1994}, and \cite{mueller_comparison_2002} is
given next:
\begin{thm}\label{thm: one}
Let $\mathcal{X}$ and $\mathcal{Y}$ be two random variables. If
$\mathcal{X} \leq_{\rm{Lt}} \mathcal{Y}$, then,
$\text{E}\left[\psi(\mathcal{X}) \right] \geq \text{E}\left[ \psi(\mathcal{Y}) \right]$ for all c.m. functions $\psi(\cdot)$, provided the expectation exists. Moreover, when $\mathcal{X} \leq_{\rm{Lt}} \mathcal{Y}$, $\text{E}[\psi(\mathcal{X})] \leq \text{E}[\psi(\mathcal{Y})]$ holds for any c.m.d. function $\psi(\cdot)$, provided the expectation exists.
\end{thm}

We will use an equivalent representation of LT ordering of discrete random variables to order the user distribution by the ordering of their PGFs. By defining $z:=e^{-s}$, one can rewrite $\text{E}\left[e^{-s\mathcal{X}}\right] \geq \text{E}\left[e^{-s\mathcal{Y}}\right]$ for $z \geq 0$ as
$\text{E}\left[z^{\mathcal{X}}\right] \geq \text{E}\left[z^{\mathcal{Y}}\right]$ for $0\leq z \leq 1$, which is the same as $U_{\mathcal{X}}(z) \geq U_{\mathcal{Y}}(z)$, $0 \leq z \leq 1$, where we recall that $U_{\mathcal{X}}(z) = \text{E}[z^\mathcal{X}]$ represents the probability generating function of the discrete random variable $\mathcal{X}$. This representation will be helpful when we compare two user distributions in Section \ref{sec: user_order}.

\subsection{Regular Variation}\label{sec: reg_vary}
A function $\psi(s)$ is {\it regularly varying} with exponent $\mu \neq 0$ at $s=\infty$ if it can be expressed as $\psi(s)= s^{\mu} l(s)$ where $l(s)$ is slowly varying which by definition satisfies $\lim_{s \rightarrow \infty} l(\kappa s)/l(s)= 1$ for $\kappa>0$. So, intuitively, regular  captures the notion of polynomial-like behavior asymptotically. Regular (slow) variation of $\psi(s)$ at $s=0$ is equivalent to regular (slow) variation of $\psi(1/s)$ at $\infty$. The Tauberian theorem for Laplace transforms, applies to c.m. functions of the form \eqref{eqn: math_prelim_two} and states that $\tau(x)$ is regularly varying at $x=\infty$ if and only if $\psi(s)$ is regularly varying at $s=0$. The following theorem is from \cite[pp. 73]{feller_introduction_2009}:

\begin{thm} \label{tbr_th}
If a non-decreasing function $\psi(s) \geq 0$ defined on $\mathbb{R}^+$
has a Laplace transform $\tau(x) = \int_0^{\infty} e^{-sx} d\psi(s)$
for $x \geq 0$, then $\psi(s)$ having variation exponent $\mu$ at $\infty$ (or $0$) and $\tau(x)$ having variation exponent $-\mu$ at $0$ (or $\infty$) imply each other.
\end{thm}

\subsection{Schur-Concave Functions and Majorization}
In this section we first introduce the notion of majorization and Shur-convex functions. For any $\bm{x} = (x_1,...,x_n) \in \mathbb{R}^n$ and $\bm{y} = (y_1,...,y_n) \in \mathbb{R}^n$, let $x_{[1]} \geq \cdots \geq x_{[n]}$ and $y_{[1]} \geq \cdots \geq y_{[n]}$ denote the components of $\bm{x}$ and $\bm{y}$ in decreasing order. We say $\bm{x}$ is majorized by vector $\bm{y}$, equivalently $\bm{x} \prec \bm{y}$ to mean $\sum_{i=1}^k x_{[i]} \leq \sum_{i=1}^k y_{[i]}$ for all $k=1,\ldots,n$, and $\sum_{i=1}^n x_{[i]}=\sum_{i=1}^n y_{[i]}$. A Schur-concave function $g$:${\mathbb{R}}^n \rightarrow \mathbb{R}$ satisfies $g(\bm{x}) \geq g(\bm{y})$ whenever $\bm{x} \prec \bm{y}$. The following theorem is proved in \cite{Ineq}:
\begin{thm}\label{thm: schur} 
Let $g$ be a continuous non-negative function defined on an interval $I \subset \mathbb{R}$. Then
\begin{align}
\phi(\bm{x}) = \prod_{i=1}^n g(x_i), \hspace{0.4 in} \bm{x} \in I^{n},
\end{align} 
is Schur-concave on $I^n$ if and only if $\log(g)$ is concave on $I$.
\end{thm}

\subsection{Bounds of Probability Generating Function}
Following theorem has been proved in \cite{PGFbound}:
\begin{thm}
Let $U_{\mathcal{N}}(z)$ be the PGF of a discrete random variable $\mathcal{N}$ with non-negative integer support. If the mean value $\lambda$ and variance $\sigma_{\mathcal{N}}^2$ exist, then the following inequalities hold for all $0\leqslant z \leqslant1$:
\begin{align}\label{eqn: PGFbound}
1+(z-1)\lambda \leqslant U_{\mathcal{N}}(z) \leqslant 1+(z-1)\lambda + \frac{(z-1)^2}{2} m(z)
\end{align}
where $m(z)/(\sigma_{\mathcal{N}}^2+{\lambda}^2-\lambda)$ is another PGF.
\end{thm}

\subsection{Asymptotics}
We say $\tau(x) = O(g(x))$ as $x \rightarrow \infty$ if and only if there is a positive constant $M$ and a real number $x_0$ such that $|f(x)| \leq M |g(x)|$ for all $x>x_0$. We say $\tau(x) = o(g(x))$ as $x \rightarrow \infty$ that for every positive integer $\epsilon$ there exists a constant $x_0$ such that $|f(x)| \leq \epsilon |g(x)|$ for all $x>x_0$ \cite{bigO}.

\section{Outage Probability}\label{sec: outprob}
The randomness of the number of active SUs arise from the selection of a desired SU according to their interference temperature at the primary receiver. Hence, how rapidly $\mathcal{N}$ varies with time depends on the rapidity fading $h_{{\rm p}_{i}}$ over the interference channel. When $h_{{\rm s}_{i}}$ and $h_{{\rm p}_{i}}$ both remain constant over the transmission duration of a codeword, the system is experiencing slow fading. Outage probability is an appropriate metric for slowly varying channels. The expression of the outage probability at average SNR $\rho$, and a desired transmit rate $R$ is defined as:
\begin{equation}\label{out_def}
\text{P}_{\rm out}(P,R) :=\text{Pr}\left[\log(1+ \rho\gamma_{\rm s}^{\ast})<R \right],
\end{equation}
where $\gamma_{\rm s}^{\ast}$ is defined in \eqref{best_gain}. Recalling that $\mathcal{N}=|\textit{S}|$, the cardinality of the active set $\textit{S}$, we can express \eqref{out_def} as:
\begin{align}\label{pgf_order}
\text{P}_{\rm out}(P,R) = \text{Pr}\left[ {\gamma_{\rm s}^{\ast}} < \frac{2^R-1}{\rho} \right] =  U_{\mathcal{N}}\left(F_{\gamma_{\rm s}}\left(\frac{2^R-1}{\rho}\right)\right)
\end{align}
using \eqref{eqn: cdf_gen}. It is clear that by comparing \eqref{pgf_order} for different user distributions, outage probability $\text{P}_{\rm out}(P,R)$ can be ordered at every value of $\rho$ and $R$ based on comparing their PGFs, also known as Laplace transform ordering. A similar property will be observed for the ergodic capacity and average BER metrics in Sections \ref{sec: erg_cap} and \ref{BER}, by using this LT ordering approach introduced in Section \ref{sec: LT_ordering}.

\section{Ergodic Capacity}\label{sec: erg_cap}
When $h_{{\rm s}_{i}}$ and $h_{{\rm p}_{i}}$ both vary rapidly over the duration of a codeword, system is in the so-called fast fading regime. We consider the ergodic capacity of the secondary system averaged over both fading and user distributions. We then study the asymptotic behavior of ergodic capacity with large mean number of SUs. The expression of the ergodic capacity of a multi-user system with deterministic number of users $N$ and average SNR $\rho$ is given by,
\begin{align}
\avcap = \int_0^{\infty} \log\left(1+\rho x \right) \dr F_{\gamma_{\rm s}}^N
(x)= \rho \int_0^{\infty} \frac{1-F_{\gamma_{\rm s}}^N (x)}{1+\rho x} \dr x .
\label{eqn: cap_av1}
\end{align}
where $\avcap$ is the ergodic capacity averaged over the fading channel. For the random number of users case, $N$ is a realization of a random variable $\mathcal{N}$, which is the number of users respecting the interference constraint. By using \eqref{eqn: cdf_gen} the ergodic capacity averaged across the user distribution can be expressed as, 
\begin{align} \label{eqn: cap1_asympt}
\ecap &= \text{E}_{{\gamma_{\rm s}^{\ast}}} [\log(1+\rho {\gamma_{\rm s}^{\ast})}] = \rho \int_0^{\infty} \frac{1-U_{\mathcal{N}}(F_{\gamma_{\rm s}}(x))}{1+\rho x} \dr x .
\end{align}

It can be shown that $\avcap$ in \eqref{eqn: cap_av1} is a c.m.d. function of $N$ \cite{6112148}. According to Theorem \ref{thm: one}, if two user distributions are LT ordered, so will their ergodic capacities. $\avcap$ is also a concave increasing function of $N$. Applying the Jensen's inequality and defining $\ld:=\text{E}[\mathcal{N}]$, we have 
\begin{equation} 
\ecap \leq \capen.
\end{equation}
Therefore, randomization of $N$ will always deteriorate the average ergodic capacity of a MUD system.

\subsection{Scaling Laws of Ergodic Capacity}\label{Ergodic_capacity}
To study how the number of active number of users $\mathcal{N}$ affects the average throughput of the system, we derive the scaling laws of the ergodic capacity for large average number of users $\ld$. 
Reference \cite{6112148} considers the Poisson distribution for $\mathcal{N}$ in a non-cognitive context and derives the scaling laws of ergodic capacity as $\ld \rightarrow \infty$. In this section, we generalize this result to a large family of user distributions and determine conditions under which similar scaling laws hold. Under a Rayleigh fading scenario, substituting $F_{\gamma_{\rm s}}(x) = 1-e^{-x}$ into \eqref{eqn: cap1_asympt} and assuming that mean value $\ld$ and variance $\sigma_{\mathcal{N}}^2$ of $\mathcal{N}$ exist, we have the following theorem:
\begin{thm}\label{thm: cogcap1}
The ergodic capacity averaged across the fading and user distribution, denoted as $\ecap$, has the following scaling law as $\ld \rightarrow \infty$, 
\begin{align} \label{eqn: cogcap1}
\ecap &=\rho \int_0^{\infty} \frac{1 - U_{\mathcal{N}}(1-e^{-x})}{1+\rho x} \dr x  \nonumber \\
& = \log\left(1+ \rho \log(\ld) \right) + O(1/{\sqrt{\log(\ld)}}),
\end{align}
\end{thm}
provided that $(a) \text{Pr}\left[ \mathcal{N}= 0 \right] =o(1/\log\log \ld)$ and $(b) \sigma_{\mathcal{N}}^2=o({\ld}^2)$ as $\ld \rightarrow \infty$.
\begin{proof}
See Appendix A.
\end{proof}
Note that if $\mathcal{N}$ is not random (i.e., the number of users is deterministic), then $\mathcal{N} = \ld$ with probability one, which satisfies both assumption $(a)$ and $(b)$ in Theorem \ref{thm: cogcap1}. This implies that $\overline{C}(\rho,L) = O(\log\log L)$, as $L \rightarrow \infty$, as also observed in \cite[Theorem 5]{5454289}. Theorem \ref{thm: cogcap1} can be viewed as a generalization of this result.

\subsection{Binomial Distributed $\mathcal{N}$}\label{sec: bino_cap}
In our proposed cognitive radio system, one possible mode of operation to select a desired user can be expressed as follows: choose the user set among $L$ total users which satisfy the interference constraint $\textit{S}=\lbrace j \in 1,\ldots, L: \gamma_{{\rm p}_j}<Q \rbrace$. Then choose the user index in $\textit{S}$ with the best channel gain $\gamma_{{\rm s}_{i}}$. In another words, the user with highest $\gamma_{{\rm s}_{i}}$ which also satisfies the interference constraint will be selected. Recall that $\mathcal{N}=|\textit{S}|$, the cardinality of $\textit{S}$, which is the number of users satisfying the interference constraint, termed as successful users. Users will be said to be failures if they are not successful. If the interference test of each user is treated as an independent Bernoulli experiment, $\mathcal{N}$ is a binomial random variable. The success probability $p$ of this binomial random variable can be represented as $F_{\gamma_{\rm p}}(Q)$, where $F_{\gamma_{\rm p}}(x)$ is the CDF of $|h_{\rm p}|^2$ which is i.i.d. across all SUs.

We use Bin($L$,$p$) to denote the binomial distribution with $L$ trials and success probability $p$. Since $\mathcal{N}$ users are chosen from $L$ total users subject to an interference threshold $Q$, the random variable $\mathcal{N}$ follows Bin($L$,$F_{\gamma_{\rm p}}(Q)$). Consequently, using \eqref{eqn: cdf_gen}, and the PGF of the binomial distribution, the CDF of the channel gain of the selected user can be expressed as,
\begin{align}
F_{{\gamma_{\rm s}^{\ast}}}(x) &= [1-F_{\gamma_{\rm p}}(Q)+F_{\gamma_{\rm p}}(Q)F_{\gamma_{\rm s}}(x)]^L \nonumber \\
&= [1-p+p(1-e^{-x})]^{\frac{\lambda}{p}}   \nonumber \\
&= (1-pe^{-x})^{\frac{\lambda}{p}}
\end{align}
where $p :=F_{\gamma_{\rm p}}(Q)$ and $\lambda :=Lp$ is the mean value of random variable $\mathcal{N}$. It can be verified that in this case $\pr[\mathcal{N}=0] = \binom{L}{0} p^0(1-p)^{\frac{\ld}{p}}$ and $\sigma_{\mathcal{N}}^2 = \ld (1-p)$, which implies that $(a)$ and $(b)$ in Theorem \ref{thm: cogcap1} are satisfied. Therefore, \eqref{eqn: cogcap1} holds for the binomial case.

\subsection{Negative Binomial Distributed $\mathcal{N}$}\label{sec: neg_bino_cap}
The number of SUs could follow discrete distributions other than binomial if different modes of operation are adopted. In the binomial case, the primary receiver performs an exhaustive search to find all active SUs among $L$ total users. When $L$ is large, this approach might require a long processing time to form the active SUs set $\textit{S}$. We term the processing time as system delay, which is in proportion to the number of SUs which has been checked for interference constraint. An alternative is to decrease the system delay by selecting the desired user from a proper subset among all users whose interference are below the threshold. 

For example, the BS can form the set $\textit{S}$ sequentially as follows. The BS selects all the active users before a predetermined number $r$ failures occurs. In this case, $\mathcal{N}$ is NB distributed with parameter $r$ and $p$, which is denoted as NB($r$,$p$). There exists a trade-off between the time BS takes to form the set $\textit{S}$ and the secondary link performance, which can be balanced by the parameter $r$. In this case, system delay is a random variable and its mean value is in proportion to $r$.

CDF of the channel gain of the best user selected from a NB random set of users can be written as using \eqref{eqn: cdf_gen} as:
\begin{align}\label{eqn: cdf_nb}
F_{{\gamma_{\rm s}^{\ast}}}(x) = \frac{1}{(1+e^{-x}u)^r},
\end{align}
where $u :=F_{\gamma_{\rm p}}(Q)/(1-F_{\gamma_{\rm p}}(Q))$, $r :=\ld/u$. Similar to the binomial $\mathcal{N}$, the conditions of Theorem \ref{thm: cogcap1} are satisfied since $\pr[\mathcal{N} = 0] = {\binom{r-1}{0}}p^0(1-p)^r$ and $\sigma_{\mathcal{N}}^2 = \ld/(1-p)$, hence \eqref{eqn: cogcap1} also holds in the NB case.

\subsection{Poisson-Binomial Distribution}\label{sec: Poisbino}
In practical systems, SUs might not necessarily suffer an interference probability that is identical across all users. Therefore, the case where SUs have different $F_{\gamma_{{\rm p}_i}}(Q)$ is of interest. In this case, the number of active SUs follows a PB distribution, which is mathematically defined as the sum of non identically distributed independent Bernoulli random variables $X_i$ so that 
\begin{equation}
\pr \left[ X_i=1\right] = p_i = 1 - \pr\left[ X_i=0\right] > 0, \hspace{0.4 in} i= 1,...,L.
\end{equation}
Let $\mathcal{W}=\sum_{i=1}^{L} X_i$ be the number of the active users among total SUs, then $\mathcal{W}$ will have a PB distribution. It is verified in Appendix C that condition $(a)$ and $(b)$ are also satisfied in this case, so that \eqref{eqn: cogcap1} holds. Furthermore, this user distribution will be studied in Section \ref{sec: poisapprox} and approximated by the Poisson distribution when $F_{\gamma_{p_i}}(Q)$ is small and all $X_i$ are independent.

\section{Average Bit Error Rate}\label{BER}
Average error rate is another key performance metric. The error rate at average SNR $\rho$ averaged over the fading and users distribution is given by
\begin{equation} \label{eqn: avg_BER1}
\text{E}_{\mathcal{N}} \left[\overline{\text{P}}_{\rm e}(\rho, \mathcal{N}) \right] = \text{E}_{\mathcal{N}} \left[ \int_0^{\infty} \berx \dr F_{\gamma_{\rm s}}^N (x) \right]
\end{equation}
where $\berx$ is the instantaneous error rate over an AWGN channel for an instantaneous SNR $\rho x$ of the best user. $\berx$ is often approximated to have the form of $\berx = \al e^{-\eta \rho x}$, where $\alpha$ and $\eta$ can be chosen to capture different modulation schemes. Other variations such as $\berx = \alpha Q(\sqrt{\eta \rho x})$ is also adopted in literature \cite{2}.
 
To see that $\avber$ is a c.m. function in $N$, consider the $k^{th}$ derivative
\begin{equation}\label{eqn: ber_dr1}
\frac{\du^{k}\avber}{\du N^{k}} = \rho \int_0^{\infty} B(\rho x)
F_{\gamma_{\rm s}}^N(x) \left[\log\left(\cdfrn\right)\right]^k \dr x,
\end{equation}
where we define $B(x)= -\dr \text{P}_e(x)/\dr x$. Since $\berx$ is decreasing in $x$ for any $\rho>0$ and $\log\left(\cdfrn \right) \leq 0$, the derivative in \eqref{eqn: ber_dr1} alternates in sign as $k$ incremented and satisfies the definition in \eqref{eqn: math_prelim_one}. Consequently, $\avber$ is a c.m. function of $N$. In Section \ref{sec: user_order}, this c.m. property along with Theorem \ref{thm: one} will be used to show that stochastic order on a pair of user distributions can be shown to order the average bit error rate under those user distributions. In particular, $\avber$ being a c.m. function of $N$ means that \eqref{eqn: ber_dr1} is negative for $k=1$ and positive for $k=2$, and consequently $\avber$ is a convex decreasing function of $N$. For the case that the number of users in the system is random, by applying Jensen's inequality, we have,
\begin{equation}\label{eqn: ber12}
\eber \geq \beren,
\end{equation}
where $\ld := \text{E}[\mathcal{N}]$. Therefore, randomization of the number of users always deteriorates the average error rate performance of a multiple SUs cognitive radio systems. In Section \ref{sec: jensen} we will show that the Jensen's inequality in \eqref{eqn: ber12} is tight for large $\ld$ and Poisson $\mathcal{N}$.

\subsection{Binomial Distributed $\mathcal{N}$}\label{sec: bino_ber}
In Section \ref{sec: bino_cap}, we derived the CDF of the channel gain of the best user chosen from a binomial distributed random set of users. Here we take derivative of \eqref{eqn: cdf_gen} with respect to $x$ so that the PDF of the channel gain of the best user in the binomial case can be expressed as:
\begin{align}\label{eqn: binpdf}
f_{{\gamma_{\rm s}^{\ast}}}(x) = \frac{\dr F_{{\gamma_{\rm s}^{\ast}}}(x)}{\dr x} = \lambda e^{-x}(1-e^{-x}p)^{\frac{\lambda}{p}-1}, \hspace{6mm} \quad x >0.
\end{align}
where we recall that $p:=F_{\gamma_{\rm p}}(Q)$.
Assuming the instantaneous error rate has the form $\berx = \alpha e^{-\eta \rho x}$ , substituting \eqref{eqn: binpdf} into \eqref{eqn: avg_BER1} we get:
\begin{align}\label{eqn: binber}
\eber &= \int_0^{\infty} \alpha e^{-\eta \rho x}e^{-x}(1-e^{-x}p)^{\frac{\lambda}{p}-1}\dr x \nonumber \\
&=\alpha p^{-1-\eta \rho} \lambda \beta\left(p,1+\eta \rho, \frac{\lambda}{p}\right)
\end{align}
where the incomplete beta function is defined as $\beta\left(x,a,b \right)=\int_0^x y^{a-1}(1-y)^{b-1} \dr y$. Note that when $p=1$ in \eqref{eqn: binber}, every SU satisfies the interference constraints, in which case $\mathcal{N}$ is deterministic. In this specific case, \eqref{eqn: binber} equals $\alpha \lambda B(1+\eta \rho,\lambda)$, which can be shown as the average BER under deterministic number of active users. Here $B(1+\eta \rho,\lambda)=\beta\left(1,1+\eta \rho, \lambda\right)$ is the beta function.

\subsection{Negative Binomial Distributed $\mathcal{N}$}\label{sec: neg_bino_ber}

In Section \ref{sec: neg_bino_cap}, we derived the CDF of the channel gain of the best user chosen from a NB distributed set of users. The PDF of the channel gain of the best user in the NB case can be expressed as:
\begin{align}\label{eqn: nbpdf}
f_{{\gamma_{\rm s}^{\ast}}}(x) = \frac{\dr F_{{\gamma_{\rm s}^{\ast}}}(x)}{\dr x} = rue^{-x}(1+ue^{-x})^{-1-r}, \hspace{6mm} \quad x >0.
\end{align}
where $r$ is the parameter of the NB distribution and $u=p/(1-p)$.
Assuming that the instantaneous error rate has the form $\berx = \alpha e^{-\eta \rho x}$, substituting \eqref{eqn: nbpdf} into \eqref{eqn: avg_BER1} we can get:
\begin{align}
\eber &= \int_0^{\infty} \alpha e^{-\eta \rho x}rue^{-x}(1+ue^{-x})^{-1-r} \dr x \nonumber \\
&=\frac{ru\alpha}{1 + \eta \rho} {}_2F_1\left( 1 + r, 1 + \eta \rho, 2 + \eta \rho, -u\right)
\end{align}
where ${}_2F_1\left(a,b,c,z\right)$ is Gauss's hyper geometric function. 
As number of failures $r$ is incremented, average BER performance improves. However, for increased $r$, the time BS takes to form set $\textit{S}$ will also be increased, so that one can balance the performance and delay trade off by adjusting the $r$ parameter.

\section{Non-homogeneous Interference Probability and Poisson Approximation}\label{sec: poisapprox}
We have introduced in Section \ref{sec: System_model} that the interference test of each SU is treated as an independent Bernoulli experiment with success probability $F_{\gamma_{\rm p}}(Q)$. In this section, the interference model will be generalized to the non-i.i.d case, in which the number of active SUs results in a PB distribution following the definition in Section \ref{sec: Poisbino}. Since it is mathematically complicated to calculate the ergodic capacity and average BER of the SU system in this case, a Poisson approximation will be utilized to approximate PB distribution.

\subsection{Poisson Approximation}
In this section, we will bound the error between the ergodic capacity under Poisson and PB $\mathcal{N}$ to show that as the PB distribution converges to Poisson distribution, the ergodic capacity $\ecap$ under PB $\mathcal{N}$ also converges to the ergodic capacity at the Poisson case.

Following the definition in Section \ref{sec: Poisbino} then $\mathcal{W}$ will have a distribution that is approximately Poisson with mean $\ld=\sum_{i=1}^{L} p_i$. This approximation will hold if $F_{\gamma_{{\rm p}_i}}(Q)$ is small and all $X_i$ are independent. We will now make this rigorous and bound the error between the ergodic capacity under a PB user distribution $\mathcal{W}$ and its corresponding Poisson approximated $\mathcal{N}$ \cite{la_cam}. First, consider the following theorem by Le Cam \cite{la_cam}:
\begin{thm}\label{thm: pois_error1}
$X_1, \ldots, X_i$ are independent random variables, each with a Bernoulli distribution of parameter $p_i$. $\text{Pr}\left[ X_i=1\right] = p_i$ for all $i= 1, \ldots,L$, i.e. $\mathcal{W}=\sum_{i=0}^{\infty} X_i$ approximately follows a PB distribution. We have
\begin{equation}\label{eqn: error_dist}
\sum_{k=0}^{\infty} \left| \text{Pr}\left[\mathcal{W}=i\right] - \frac{e^{-\ld}{\ld}^{i}}{i!} \right| \leq 2\sum_{i=1}^{L} {p_i}^2, \hspace{6mm} \quad i=1,2,\ldots,L
\end{equation}
where $\ld=\sum_{i=1}^{L}p_i$.
\end{thm}

Using Theorem \ref{thm: pois_error1}, we will bound the gap between the ergodic capacity under PB and Poisson distributions, which is denoted as $\Delta_C$. We have:
\begin{align}
\Delta_{C} &= \left| \text{E}_{\mathcal{W}} \left[\overline{C}(\rho,\mathcal{W}) \right]-\text{E}_{\mathcal{N}}\left[\overline{C}(\rho,\mathcal{N}) \right] \right| \nonumber \\
&= \left| \sum_{i=1}^{L} \overline{C}(\rho,i)\left( \text{Pr}\left[\mathcal{W}=i\right] - \frac{e^{-\ld}{\ld}^{i}}{i!} \right) \right| \nonumber \\
&\leq \sum_{i=1}^{L} \overline{C}(\rho,i) \left| \left( \text{Pr}\left[\mathcal{W}=i\right] - \frac{e^{-\ld}{\ld}^{i}}{i!} \right) \right|. \nonumber
\end{align}
Since $\overline{C}(\rho,i)$ is increasing in $i$, and $\overline{C}(\rho,L) = O\left(\log\log L \right)$ as we mentioned in Section \ref{Ergodic_capacity}, applying \eqref{eqn: error_dist} we have:
\begin{align}
\Delta_{C} &= O\left(\log\log L \sum_{i=1}^{L} {p_i^2} \right) 
\end{align}
where $i=1,2,\ldots,L$. As long as $\sum_{i=1}^{L} {p_i^2} = o(1/(\log\log L))$, the error between the capacity under PB and Poisson distributions goes to zero as $L \rightarrow \infty$.

For a special case consider $p_i=\ld/L$ for $i=1,2,\ldots,L$, in which all the SUs have i.i.d. interference channels, $\mathcal{W}$ follows a binomial distribution. In this case, we have
\begin{align}
\Delta_{C} = O \left(\frac{\ld^2}{L} \log\log L \right)
\end{align}
Obviously, as $L \rightarrow \infty$ and $p \rightarrow 0$, $\Delta_{C}$ approaches zero. Consequently, the gap between the binomial and the approximated Poisson capacity is shown to be negligible as total number of users grows large and the interference probability is sufficiently small. This will be illustrated numerically in Section \ref{sec: simulations}.

\subsection{Tightness in the Jensen's Inequality in the Average BER}\label{sec: jensen}
Since in Section \ref{sec: poisapprox} we proved that Poisson distribution can be utilized to precisely approximate PB distribution, it is of interest to study the average BER under Poisson $\mathcal{N}$. In Section \ref{BER}, we proved that the average BER $\avber$ is a completely monotonic function of $N$, which implies the convexity. Applying the Jensen's inequality, we have \eqref{eqn: ber12}.

We now provide sufficient conditions for Jensen's inequality involving $\avber$ to be asymptotically tight for large $\ld$. Recall that $\avber$ is the error rate averaged over the channel distribution for deterministic number of users $N$. To this end, we use \cite[Theorem 2.2]{Downey93anabelian} which were derived in a networking context for arbitrary {\it c.m.} functions.

\begin{thm}\label{thm: four}
Let $\avber$ be c.m. and regularly varying at $N=\infty$ and consider the error rate averaged across the channel and the users $\eber$, where $\mathcal{N}$ is a Poisson distributed random variable with mean $\ld$. Then, 
\begin{equation}\label{eqn: jsens_tight}
\eber = \beren+O\left(\overline{\text{P}}_{\rm e}(\rho, \ld)/\ld \right)
\end{equation}
as $\ld \rightarrow \infty$.
\end{thm}

Equation \eqref{eqn: jsens_tight} shows that as $\ld \rightarrow \infty$, the difference between the error rate averaged across the user distribution and the error rate evaluated at the average number of users vanishes as $\ld$ tends to $\infty$. This implies that for sufficiently large $\ld$ the performance of the MUD systems with random number of users will be almost equal to the performance of the MUD systems with a deterministic number of users with the number of users equal to $\ld$.

To apply Theorem \ref{thm: four} we require $\avber$ to be c.m. and regularly varying. We have already shown in Section \ref{BER} that $\avber$ is always completely monotonic in $N$. Next, we provide the conditions under
which $\avber$ is a regularly varying function of $N$. Consider 
\begin{equation}
\avber = \rho \int_0^{\infty} B(\rho x) e^{N \log \cdfrn} \dr x
\end{equation}
where $B(\cdot)$ is defined as $B(x)= -\dr \text{P}_e(x)/\dr x$.
Now, setting $u := -\log(\cdfrn)$, and integrating by substitution we have,
\begin{equation} \label{eqn: downey_res3}
\avber = \rho \int_0^{\infty} \frac{B(\rho \cdfinv) e^{-u} e^{-u N}
\dr u}{f_{\gamma_{\rm s}}(\cdfinv)},
\end{equation}
where $F_{\gamma_{\rm s}}^{-1}(x)$ is the inverse CDF and $f_{\gamma_{\rm s}}(x)$ is the PDF of $\gamma_{\rm s}$. We now establish the sufficient conditions for $\avber$ to be a regularly varying function of $N$:

\begin{thm}\label{thm: five}
If $\avber$ is c.m. in $N$, a sufficient condition for it to be
regularly varying at $N=\infty$ is that, $t(u):=\rho (B(\rho
\cdfinv)e^{-u})/(f_{{\gamma_{\rm s}}}\left(\cdfinv\right))$ is regularly
varying at $u=0$.
\end{thm}

\begin{proof}
By comparing the representation of $\avber$ in \eqref{eqn: downey_res3} with the Bernstein's representation of c.m. functions discussed after \eqref{eqn: math_prelim_one}, it can be seen that
\eqref{eqn: downey_res3} can be represented as the Laplace transform of $t(u)$. Using Theorem \ref{tbr_th}, the proof follows.
\end{proof}

Theorem \ref{thm: five} shows that for the conclusions of Theorem \ref{thm: four} to hold (i.e., Jensen's inequality to be asymptotically tight), the CDF of the single-user channel $F_{\gamma_{\rm s}}(x)$, and the error rate expression $\berx$ have to jointly satisfy the regular variation condition given in Theorem \ref{thm:
five}. Next, we examine whether this condition holds for commonly assumed instantaneous error rates $\berx$ with $\gamma_{\rm s}$ being exponentially distributed. For the case of $\berx= \alpha e^{-\eta
\rho x}$, we have $t(u)= \alpha \rho (1-\eu)^{\eta \rho-1} \eu$, which satisfies $\lim_{u \rightarrow 0} t(\kappa u)/t(u) = \kappa^{\eta \rho-1}$, therefore proving the regular variation of $t(u)$ at $0$. By using Theorem \ref{tbr_th} this in turn proves regular variation of $\avber$ at $N=\infty$. Therefore $\avber$ is both a c.m. and a regularly varying function of $N$ for this case. Consequently, when $\berx= \alpha e^{-\eta \rho x}$ and the fading is Rayleigh (i.e. channel gain is exponential), the difference in error rate performance of a MUD system with a random number of users averaged over the number of users distribution and of a deterministic number users approaches zero for sufficiently large
$\ld$, as in Theorem \ref{thm: four}.

Consider now $\berx = \alpha Q(\sqrt{\eta \rho x})$, with $\gamma_{\rm s}$ being exponentially distributed. The error rate can be expressed as,
\begin{equation} \label{eqn: ex_exp1}
\avber = \alpha \int_0^{\infty} Q\left(\sqrt{\eta \rho x}\right)\dr
F_{\gamma_{\rm s}}^N(x) = \frac{\alpha \sqrt{\eta \rho}}{2 \sqrt{2\pi}}
\int_0^{\infty} \frac{e^{N\log\left(1-e^{-x}\right)} e^{-\eta \rho
x/2}}{\sqrt{x}} \dr x,
\end{equation}
where the second equality is obtained by integration by parts. Once again, by setting $u = -\log(1-e^{-x})$ we can rewrite \eqref{eqn: ex_exp1} as,
\begin{equation}
\frac{\alpha \sqrt{\eta \rho}} {2\sqrt{2\pi}} \int_0^{\infty}
\exp\left(-Nu\right) (1-e^{-u})^ {\eta \rho/2-1}
\frac{e^{-u}}{\sqrt{-\log(1-e^{-u})}} \dr u.
\end{equation}
Thus we have $t(u)=\alpha \sqrt{\eta \rho} (1-e^{-u})^{\eta \rho/2-1}e^{-u}/(2\sqrt{-2\pi \log(1-e^{-u})})$ and it can be shown that $\lim_{u \rightarrow 0} t(\kappa u)/t(u) = \kappa^{\eta \rho/2-1}$, therefore once again proving that $\avber$ is both a c.m. and a regularly varying function of $N$. Having verified the conditions of Theorem \ref{thm: five} for $\berx = \alpha Q(\sqrt{\eta \rho x})$ with $\gamma_{\rm s}$ being exponentially distributed, we conclude the tightness of Jensen's inequality as suggested by Theorem \ref{thm: four}.

\section{Laplace Transform Ordering of User Distributions}\label{sec: user_order}
We know from Jensen's inequality that a deterministic number of SUs will always outperform a random number of SUs both for average BER and ergodic capacity. Moreover, different random SU distributions can also be ordered among themselves. In this section, we introduce Laplace transform (LT) ordering, a method to compare the effect that different user distributions has on the average error rate, ergodic capacity, or other metrics that are either c.m. or c.m.d. in the number of active users. From \cite{6112148} we know that ergodic capacity is c.m.d. and averaged BER is c.m.. Consequently, Theorem \ref{thm: one} implies that if the number of users is from a distribution that can be ordered in the LT sense, then both the average error rate and capacity can be ordered at every value of SNR $\rho$.

\begin{thm}\label{thm: binopois}
Let $\mathcal{X}$ denote a Poisson random variable with parameter $\lambda$, $\mathcal{Y}$ denotes a binomial random variable with mean value $Lp$, and $\mathcal{Z}$ denote a NB random variable with mean value $rp/(1-p)$, and $\mathcal{W}$ denote a PB random variable defined in Theorem \ref{thm: pois_error1}. By assuming equal mean for all distributions, that is $Lp=\lambda=rp/(1-p)=\sum_{i=1}^L p_i$, we have $U_{\mathcal{Z}}(z) \geq U_{\mathcal{X}}(z)\geq U_{\mathcal{Y}}(z) \geq U_{\mathcal{W}}(z)$, for $0 \leq z \leq 1$. In other words
\begin{equation} \label{eqn: LTorderbino}
\mathcal{Z} \leq_{\rm Lt} \mathcal{X} \leq_{\rm Lt} \mathcal{Y} \leq_{\rm Lt} \mathcal{W}.
\end{equation}
\end{thm}

\begin{proof}
See Appendix B.
\end{proof}

It can be observed that for the extreme case that when parameter $p=1$, binomial user distribution converges to the deterministic number of users, which dominates any kind of random distributions with the same mean value under LT ordering sense. PB user distribution also subsumes deterministic case when $p_i$ are either $1$ or $0$. Moreover, due to Theorem \ref{thm: one}, if the SU distributions are ordered in LT sence, any c.m. (c.m.d.) performance metric of $N$ will also be ordered. Hence, without calculating or deriving the closed form expression, system performance can be compared after knowing the corresponding user distributions.

\section{Simulations}\label{sec: simulations}
An uplink cognitive radio system with multiple SUs where both SUs and BS having a single antenna is considered. In this section, using Monte-Carlo simulations, ergodic capacity and averaged BER are simulated to corroborate our analytical results. For all simulations, Rayleigh fading channels are assumed.

\begin{center}
\includegraphics[height=9.5cm,width=11.5cm]{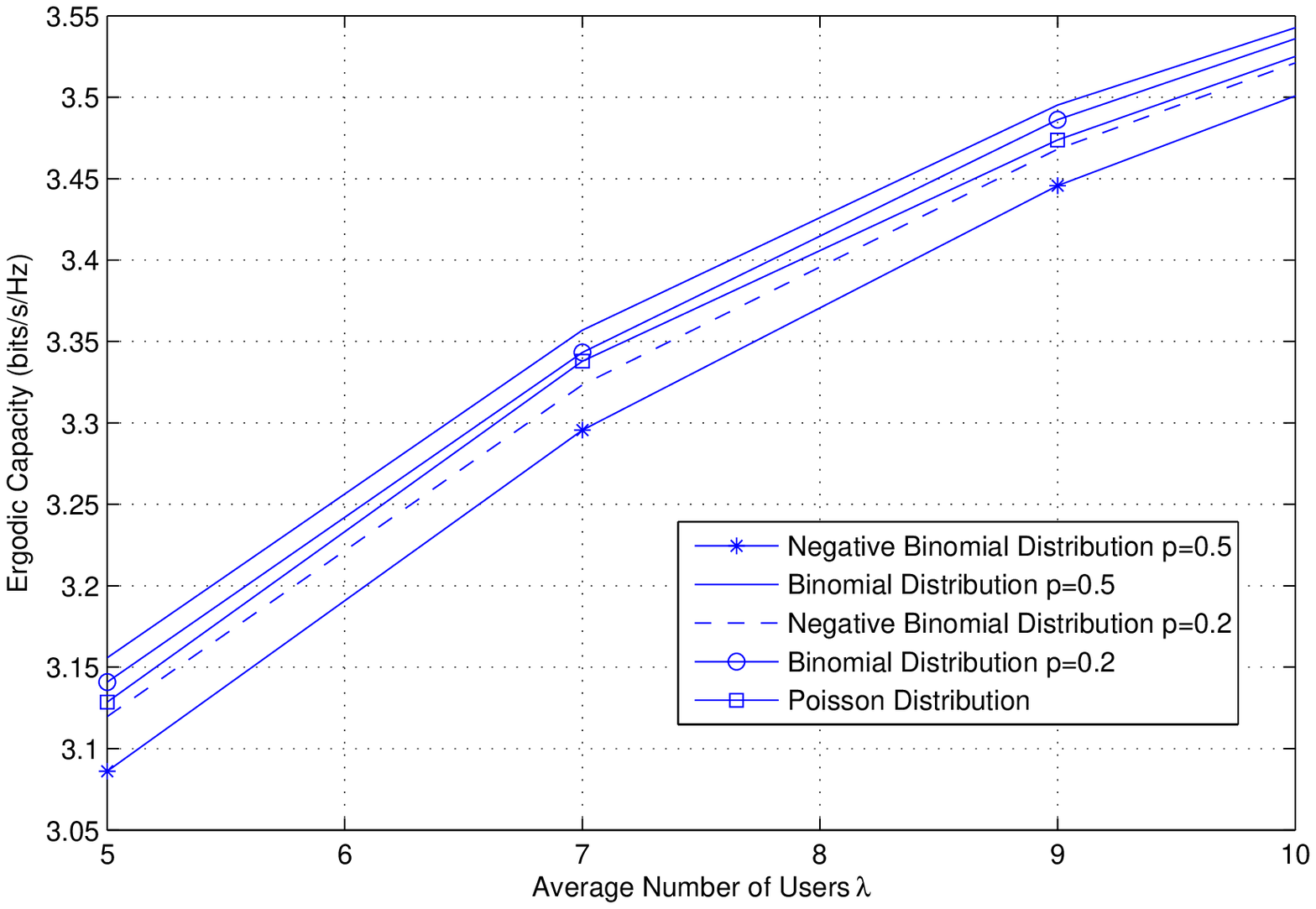}
\end{center}
\vspace{-0.3in}
\figcaption{Ergodic Capacity Under Different User Distributions.}
\label{fig: cogcomp}
In Section \ref{sec: bino_cap} and \ref{sec: neg_bino_cap}, ergodic capacity performances under binomial and NB user distributions are established. In Figure \ref{fig: cogcomp}, ergodic capacity is plotted versus $\ld=\text{E}[\mathcal{N}]$ for different user distributions. It can be seen that for a given user distribution, the ergodic capacity improves with average number of users. Also, in Section \ref{sec: user_order}, these two distributions are compared with the Poisson distribution in LT ordering sense. In Figure \ref{fig: cogcomp}, for a given $\ld$, binomial user distribution yields better ergodic capacity performance than Poisson, followed by NB user distribution. Furthermore, NB distribution converges to the Poisson distribution as the trial probability $p \rightarrow 0$ and stopping parameter $r \rightarrow \infty$, and the binomial distribution also converges towards the Poisson distribution as the number of trials goes to infinity and the product $Lp$ remains fixed. It can be seen from Figure \ref{fig: cogcomp} that for a fixed $\ld$, when the trial probability $p$ varies from 0.5 to 0.2, ergodic capacity of binomial and NB cases converge to the Poisson case.

\begin{center}
\includegraphics[height=9.5cm,width=11.5cm]{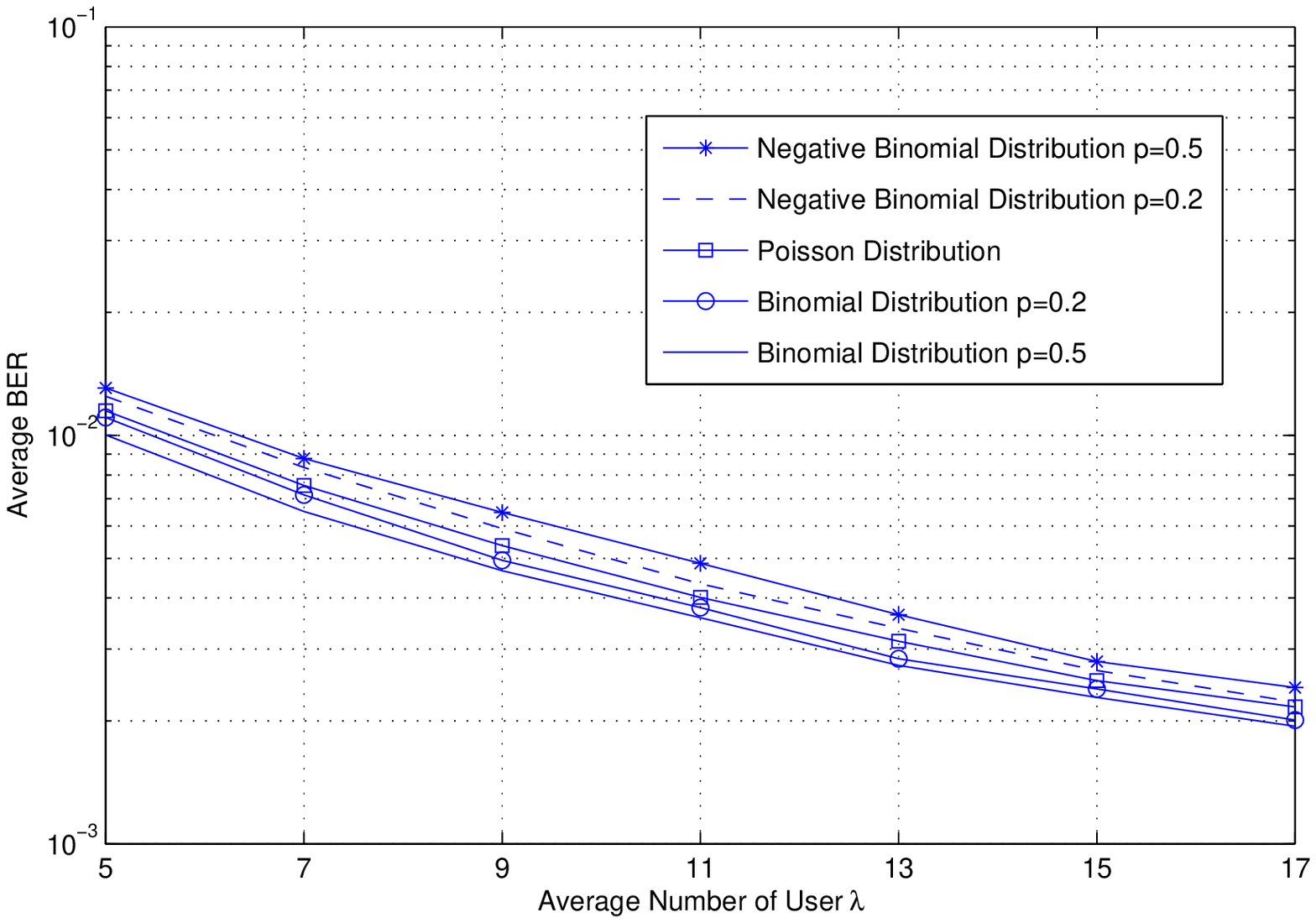}
\end{center}
\vspace{-0.3in}
\figcaption{Average BER Under Different User Distributions.}
\label{fig: cogbercomp}
In Section \ref{sec: bino_ber} and \ref{sec: neg_bino_ber}, we derived closed form expressions for averaged BER under binomial and NB user distributions. As we introduced in Section \ref{sec: LT_ordering}, $\avber$ is c.m. in $N$ and if $\mathcal{Z} \leq_{\rm Lt} \mathcal{X} \leq_{\rm Lt} \mathcal{Y}$, we have $\text{E}_{\mathcal{Z}} \left[\overline{\text{P}}_{\rm e}(\rho, \mathcal{Z}) \right] \geq \text{E}_{\mathcal{X}} \left[\overline{\text{P}}_{\rm e}(\rho, \mathcal{X}) \right] \geq \text{E}_{\mathcal{Y}} \left[\overline{\text{P}}_{\rm e}(\rho, \mathcal{Y}) \right]$, $\forall \rho$. As shown in Figure \ref{fig: cogbercomp}, average BER performances under are plotted against $\ld=\text{E}[\mathcal{Z}] = \text{E}[\mathcal{X}] = \text{E}[\mathcal{Y}]$. Here, $\mathcal{Z}$ is NB, $\mathcal{X}$ is Poisson, and $\mathcal{Y}$ is binomial distributed. With same mean number of users, average BER under binomial $\mathcal{Y}$ always outperforms Poisson $\mathcal{X}$ and NB $\mathcal{Z}$. Additionally, when $\ld$ is fixed, BER performance under binomial and NB cases converge to Poisson case as $p$ is decreased from 0.5 to 0.2.

\begin{center}
\includegraphics[height=9.5cm,width=11.5cm]{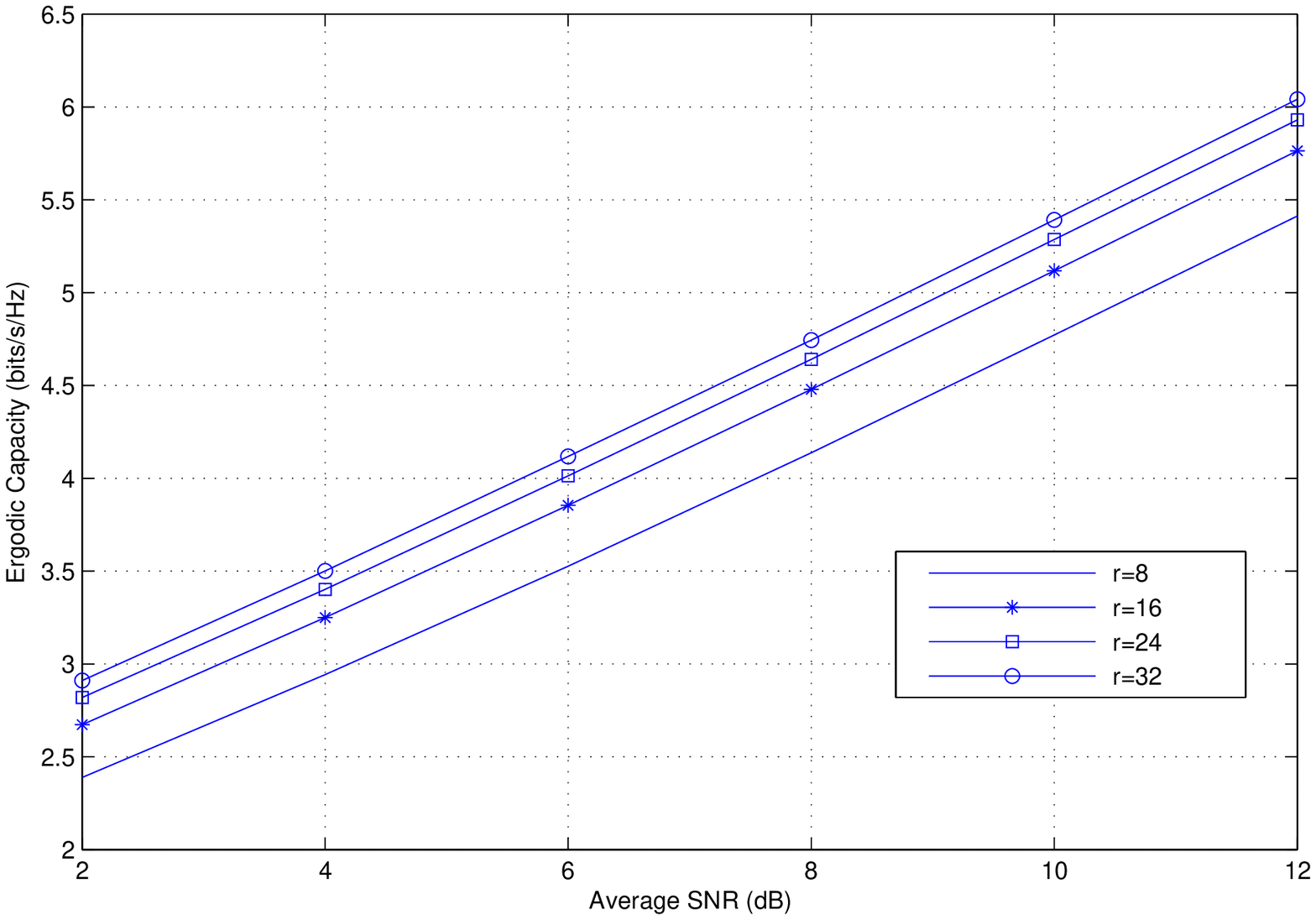}
\end{center}
\vspace{-0.3in}
\figcaption{Ergodic Capacity of NB $\mathcal{N}$ with Different $r$.}
\label{fig: delayperformance}
In section \ref{sec: neg_bino_cap}, we mentioned that in NB case, the trade-off between performance and system delay can be balanced by a choice of the parameter $r$. As shown in Figure \ref{fig: delayperformance}, for a given average SNR, ergodic capacity under different values of parameter $r$ are simulated. It can be observed that as $r$ increases from 8 to 32, the ergodic capacity performance is increased only approximately 20$\%$. However, the average system delay when $r=32$ is four times as much as $r=8$. Hence there are diminishing returns in capacity as the delay parameter is increased.

\section{Conclusions}\label{sec: conclusions}
An underlay cognitive radio system with multiple SUs is analyzed when the number of SUs is random and MUD is used. Outage probability is related directly to the probability generating function of the number of active users. The scaling laws of the ergodic capacity for large mean number of users are studied. Closed form non-asymptotic expressions for the averaged bit error rate under binomial, and NB active users are also derived. A non-homogeneous interference scenario is also considered where the number of active SUs follows PB distribution. In this case, Poisson approximation is applied to study the ergodic capacity performance. Furthermore, outage probability, ergodic capacity, averaged BER performances for two different user distributions are shown to be ordered if the user distributions are LT ordered.

\appendices
\section{Proof of Theorem \ref{thm: cogcap1}}
Defining $y := e^{-x}$ and integrating by substitution,
\begin{align}
\ecap &= \rho \int_0^\infty \frac{1-U_{\mathcal{N}}(1-e^{-x})}{1+\rho x} \dr x \nonumber \\
&= \int_0^1 \frac{1-U_{\mathcal{N}}(1-y)}{1-\rho \log y} \left(\frac{\rho}{y} \right) \dr y \nonumber \\
&= \int_0^{\frac{\sqrt{\log \ld}}{\ld}}  \frac{1-U_{\mathcal{N}}(1-y)}{1-\rho \log y}
\left(\frac{\rho}{y} \right) \dr y + \int_{\frac{\sqrt{\log \ld}}{\ld}}^1
\frac{1-U_{\mathcal{N}}(1-y)}{1-\rho \log y} \left(\frac{\rho}{y} \right) \dr y 
\label{eqn: cog_pexp1}
\end{align}
For the first term after the third equality in \eqref{eqn: cog_pexp1}, we have the following inequalities according to the lower bound in \eqref{eqn: PGFbound}:
\begin{align}
0 &< \int_0^{\frac{\sqrt{\log \ld}}{\ld}} \frac{U_{\mathcal{N}}(1-y)}{1-\rho \log y} \left(\frac{\rho}{y} \right) \dr y <
\int_0^{\frac{\sqrt{\log \ld}}{\ld}} \frac{\ld y}{1+\rho \log(\lambda)} \left(\frac{\rho}{y} \right)  \dr y \label{numerator} \\
&= \frac{\rho \sqrt{\log(\ld)}}{1+\rho \log(\lambda)-(\rho/2)\log(\log(\ld))}   \label{cog_term1_bound}
\end{align}
The right hand side in \eqref{numerator} holds because the denominator of the integrand is replaced with its lower limit. It can be seen that the upper bound after the equality in
\eqref{cog_term1_bound} yields $O \left(1/\sqrt{\log(\ld)}\right)$ and has limit $0$ as $\ld \rightarrow \infty$. This implies that the first term in \eqref{eqn: cog_pexp1} should have limit $0$. The second term in \eqref{eqn: cog_pexp1} has the upper and lower bounds given by,
\begin{align}\label{eqn: dominantterm}
&\int_{\frac{\sqrt{\log \ld}}{\ld}}^1 \frac{1-
U_{\mathcal{N}}\left(1-\frac{\sqrt{\log(\ld)}}{\ld}\right)}{(1-\rho \log(y))} \left(\frac{\rho}{y} \right) \dr y \nonumber \\
&< \int_{\frac{\sqrt{\log \ld}}{\ld}}^1 \frac{1-U_{\mathcal{N}}(1-y)}{1-\rho \log (y)} \left(\frac{\rho}{y} \right) \dr y \nonumber \\ 
&< \int_{\frac{\sqrt{\log \ld}}{\ld}}^1 \frac{1-U_{\mathcal{N}}(0)}{1-\rho \log (y)} \left(\frac{\rho}{y} \right) \dr y 
\end{align}
in which the lower and upper bounds are obtained by bounding the numerator since $U_{\mathcal{N}}(1-y)$ is a monotonically decreasing function of $y$. Defining a normalized random variable $\mathcal{N}^{\prime} = \mathcal{N}/\ld$ which has mean value 1 and variance ${\sigma_{\mathcal{N}^{\prime}}^2} = o(1)$ as $\ld \rightarrow \infty$, using the upper bound in \eqref{eqn: PGFbound} we have:
\begin{align}
U_{\mathcal{N}^{\prime}}(s) &\leqslant 1 - (1-s) + \frac{\sigma_{\mathcal{N}^{\prime}}^2}{2}(1-s)^2 \nonumber \\
&=s + \frac{\sigma_{\mathcal{N}^{\prime}}^2}{2}(1-s)^2
\end{align}
and
\begin{align}
U_{\mathcal{N}}(s)=U_{\mathcal{N}^{\prime}}(s^{\ld}) \leqslant s^{\ld} + \frac{\sigma_{\mathcal{N}^{\prime}}^2}{2}(1-s^{\ld})^2.
\end{align}
Moreover, the numerator of the lower bound in \eqref{eqn: dominantterm} can be further lower bounded as following:
\begin{align}
1-U_{\mathcal{N}}\left(1-\frac{\sqrt{\log(\ld)}}{\ld}\right) \geqslant 1- \left(1-\frac{\sqrt{\log(\ld)}}{\ld}\right)^{\ld} - \frac{\sigma_{\mathcal{N}^{\prime}}^2}{2}\left(1-\left(1-\frac{\sqrt{\log(\ld)}}{\ld}\right)^{\ld}\right)^2
\end{align}
Therefore, the upper and lower bounds in \eqref{eqn: dominantterm} turn out to be 
\begin{align}\label{eqn: 11}
\left(1 - g(\ld) - \frac{\sigma_{\mathcal{N}^{\prime}}^2}{2}\left(1-g(\ld)\right)^2 \right) \log\left(1 + \rho\log(\ld) - \frac{\rho}{2}\log(\log(\ld))\right)
\end{align}
and 
\begin{align}\label{eqn: 22}
\left(1-U_{\mathcal{N}}(0) \right)\log\left(1 + \rho\log(\ld) - \frac{\rho}{2}\log(\log(\ld))\right)
\end{align}
respectively, where $g(\ld)=\left(1-\frac{\sqrt{\log(\ld)}}{\ld}\right)^{\ld}$. In \eqref{eqn: 11}, when condition $(b)$ holds,
\begin{align}\label{eqn: limit}
\lim_{\ld \rightarrow \infty}\left(g(\ld) + \frac{\sigma_{\mathcal{N}^{\prime}}^2}{2}\left(1-g(\ld)\right)^2\right) \log(1+\log(\ld))=0
\end{align}
and in \eqref{eqn: 22}, when condition $(a)$ holds 
\begin{align}
\lim_{\ld \rightarrow \infty}U_{\mathcal{N}}(0)\log(1 + \rho\log(\ld)) = 0.
\end{align}
Hence, both \eqref{eqn: 11} and \eqref{eqn: 22} converge to $\log\left(1 + \rho\log(\ld) - (\rho/2)\log(\log(\ld))\right)$ as $\ld \rightarrow \infty$. Moreover,
\begin{align}\label{eqn: complete}
\log\left(1 + \rho\log(\ld) - \frac{\rho}{2}\log(\log(\ld))\right) &= \log \left( \left( 1+\rho \log (\ld) \right) \left( 1- \frac{\frac{\rho \log(\log(\ld))}{2}}{\log(1+\rho \log(\ld))} \right) \right) \nonumber \\
&= \log(1+\rho \log(\ld)) + \log \left( 1- \frac{\frac{\rho \log(\log(\ld))}{2}}{\log(1+\rho \log(\ld))} \right) \nonumber \\
&=\log(1+\rho \log(\ld)) + O\left(\log(\log(\ld))/\log(\ld)\right)
\end{align} as $\ld \rightarrow \infty$. 
Therefore, considering the fact that $\log(\log(\ld))/\log(\ld)$ decays faster that $1/\sqrt{\log(\ld)}$, \eqref{cog_term1_bound} and \eqref{eqn: complete} complete the proof.

\section{Theorem \ref{thm: cogcap1} Holds For PB SU Distribution}
Following the definition in Section \ref{sec: Poisbino}, we have $\sum_{i=1}^L p_i = \ld$ and $\sigma_{\mathcal{W}}^2 = \sum_{i=1}^L p_i(1-p_i)$. Since $\pr\left[\mathcal{W}= 0\right] = \prod_{i=1}^L(1-p_i) = U_{\mathcal{W}}(0)$ and $\log(1-p_i) < -p_i$, we have
\begin{align} 
\log \left( \prod_{i=1}^L (1-p_i) \right) = \sum_{i=1}^L \log(1-p_i) < \sum_{i=1}^L(-p_i) = -\ld,
\end{align}
equivalently, $0< \pr\left[ \mathcal{W}= 0 \right] < e^{-\ld}$, which implies that condition $a$ is satisfied. Furthermore, it is obvious that $\sigma_{\mathcal{W}}^2 < \sum_{i=1}^L p_i =\ld$, hence condition $(b)$ holds.

\section{Proof of Theorem \ref{thm: binopois}}
$\mathcal{Z} \leq_{\rm{Lt}} \mathcal{X}$: to show $U_{\mathcal{Z}}(z)\leq U_{\mathcal{Y}}(z)$ first we take logarithm to $U_{\mathcal{Z}}(z)$ and $U_{\mathcal{Y}}(z)$ and we get
\begin{align}\label{eqn: PGFdif}
\log(U_{\mathcal{Z}}(z))-\log(U_{\mathcal{Y}}(z))=\frac{pr}{1-p}(z-1)-r\log\left(\frac{1-p}{1-pz}\right). 
\end{align}
By shuffling the terms we rewrite the problem as comparing $\frac{p}{1-p}-\log(1-p)+\log(1-pz)$ with $0$. Taking the first derivative with respect to $z$ we get
\begin{equation}\label{eqn: PGFder}
\frac{\du\left(\frac{p}{1-p}-\log(1-p)+\log(1-pz)\right)}{\du z} =\frac{p}{1-p}-\frac{p}{1-pz} \geq 0
\end{equation}
for all $0 \leq z \leq 1$. This implies that \eqref{eqn: PGFdif} is an monotonically increasing function of $z$ with the maximum value $0$ at $z=1$. 

$\mathcal{X} \leq_{\rm{Lt}} \mathcal{Y}$: to show $U_{\mathcal{X}}(z)\geq U_{\mathcal{Y}}(z)$ first we take logarithm to $U_{\mathcal{X}}(z)$ and $U_{\mathcal{Y}}(z)$ and we get
\begin{align}\label{eqn: PGFdifbino}
\log(U_{\mathcal{X}}(z))-\log(U_{\mathcal{Y}}(z))=Lp(z-1)-L\log\left(1-p+pz\right). 
\end{align}
By rearranging the terms we rewrite the problem as comparing $s-\log(1+s)$ with $0$, where $s=p(z-1)$. Taking the first derivative with respect to $s$ we get
\begin{align}\label{eqn: PGFderbino}
\frac{\du\left(s-\log(1+s)\right)}{\du s} &=1-\frac{1}{s+1} \nonumber \\
& = \frac{s}{1+s} \leq 0
\end{align}
for all $0 \leq z \leq 1$. This implies that \eqref{eqn: PGFdifbino} is an monotonically decreasing function of $z$ with the minimum value $0$ at $z=1$. 

$\mathcal{Y} \leq_{\rm{Lt}} \mathcal{W}$: First we express all success probabilities of $\mathcal{W}$, denoted as $p_i$, in vector form so that $\bm{p}=[p_1 \quad p_2 \quad \ldots \quad p_i]$. To show that $U_{\mathcal{Y}}(z) \geq U_{\mathcal{W}}(z)$ we notice that the equality is achieved when $\bm{p}=[\ld/L \quad \ld/L \quad \ldots \quad \ld/L]$ which is denoted as $\bm{p}_{\rm bin}$. By applying Theorem \ref{thm: schur} it can be seen that $U_{\mathcal{W}}(z)=\prod_{i=1}^L(1-p_i+p_{i}z)$ is a Shur-concave function of $\bm{p}$ since $\log(1-p_i+p_{i}z)$ is concave of $p_i$. Second, we assume there exists at least one probability vector of $\mathcal{W}$, denoted as $\bm{w}=[w_1 \quad w_2 \quad \ldots \quad w_i]$, and $\bm{w} \prec \bm{p}_{\rm bin}$. Then we have
\begin{align}
w_1 = \max_{i}w_i \leq \ld/L,
\end{align}
so that 
\begin{align}
\sum_{i=1}^L w_i \leq \sum_{i=1}^L w_1 \leq \ld
\end{align}
which will violate the condition $\sum_{i=1}^L w_i=\ld$ unless $\bm{w}=\bm{p}$. This indicates that $\bm{p}_{\rm bin}$ is majorized by any other arbitrary probability vector $\bm{w}$ of $\mathcal{W}$, and $U_{\mathcal{Y}}(z,\bm{p})= U_{\mathcal{W}}(\bm{p}_{z,\rm{bin}})\geq U_{\mathcal{W}}(z,\bm{w})$ at every value of $z$.

Hence, we have $U_{\mathcal{Z}}(z) \geq U_{\mathcal{X}}(z)\geq U_{\mathcal{Y}}(z) \geq U_{\mathcal{W}}(z)$, equivalently it can be concluded that $\mathcal{Z} \leq_{\rm Lt} \mathcal{X} \leq_{\rm Lt} \mathcal{Y} \leq_{\rm Lt} \mathcal{W}$, which completes the proof. 

\bibliographystyle{IEEEtran}

\bibliography{adarsh_ref}

\begin{thebibliography}{10}
\providecommand{\url}[1]{#1}
\csname url@samestyle\endcsname
\providecommand{\newblock}{\relax}
\providecommand{\bibinfo}[2]{#2}
\providecommand{\BIBentrySTDinterwordspacing}{\spaceskip=0pt\relax}
\providecommand{\BIBentryALTinterwordstretchfactor}{4}
\providecommand{\BIBentryALTinterwordspacing}{\spaceskip=\fontdimen2\font plus
\BIBentryALTinterwordstretchfactor\fontdimen3\font minus
  \fontdimen4\font\relax}
\providecommand{\BIBforeignlanguage}[2]{{%
\expandafter\ifx\csname l@#1\endcsname\relax
\typeout{** WARNING: IEEEtran.bst: No hyphenation pattern has been}%
\typeout{** loaded for the language `#1'. Using the pattern for}%
\typeout{** the default language instead.}%
\else
\language=\csname l@#1\endcsname
\fi
#2}}
\providecommand{\BIBdecl}{\relax}
\BIBdecl

\bibitem{788210}
J.~Mitola and G.~Maguire, ``Cognitive radio: making software radios more
  personal,'' \emph{Personal Communications, IEEE}, vol.~6, no.~4, pp. 13--18,
  1999.

\bibitem{6054064}
J.~Kim, Y.~Shin, T.~W. Ban, and R.~Schober, ``Effect of spectrum sensing
  reliability on the capacity of multiuser uplink cognitive radio systems,''
  \emph{Vehicular Technology, IEEE Transactions on}, vol.~60, no.~9, pp. 4349
  --4362, Nov. 2011.

\bibitem{4493828}
J.-H. Baek, H.-J. Oh, and S.-H. Hwang, ``Improved reliability of spectrum
  sensing using energy detector in cognitive radio system,'' in \emph{Advanced
  Communication Technology, 2008. ICACT 2008. 10th International Conference
  on}, vol.~1, Feb. 2008, pp. 575 --578.

\bibitem{6178840}
Y.~Zou, Y.-D. Yao, and B.~Zheng, ``Cooperative relay techniques for cognitive
  radio systems: Spectrum sensing and secondary user transmissions,''
  \emph{Communications Magazine, IEEE}, vol.~50, no.~4, pp. 98 --103, April
  2012.

\bibitem{5434164}
T.~Do and B.~Mark, ``Exploiting multiuser diversity for spectrum sensing in
  cognitive radio networks,'' in \emph{Radio and Wireless Symposium (RWS), 2010
  IEEE}, Jan. 2010, pp. 228 --231.

\bibitem{4840529}
A.~Goldsmith, S.~Jafar, I.~Maric, and S.~Srinivasa, ``Breaking spectrum
  gridlock with cognitive radios: An information theoretic perspective,''
  \emph{Proceedings of the IEEE}, vol.~97, no.~5, pp. 894--914, 2009.

\bibitem{4069138}
M.~Gastpar, ``On capacity under receive and spatial spectrum-sharing
  constraints,'' \emph{Information Theory, IEEE Transactions on}, vol.~53,
  no.~2, pp. 471 --487, Feb. 2007.

\bibitem{4786488}
T.~W. Ban, W.~Choi, B.~C. Jung, and D.~K. Sung, ``Multi-user diversity in a
  spectrum sharing system,'' \emph{Wireless Communications, IEEE Transactions
  on}, vol.~8, no.~1, pp. 102 --106, Jan. 2009.

\bibitem{5403611}
R.~Zhang and Y.-C. Liang, ``Investigation on multiuser diversity in spectrum
  sharing based cognitive radio networks,'' \emph{Communications Letters,
  IEEE}, vol.~14, no.~2, pp. 133 --135, February 2010.

\bibitem{6133630}
C.~Masouros, F.~Khan, T.~Ratnarajah, and M.~Sellathurai, ``On the diversity
  gains of user scheduling in the cognitive radio parallel access channel,'' in
  \emph{Global Telecommunications Conference (GLOBECOM 2011), 2011 IEEE}, Dec.
  2011, pp. 1 --5.

\bibitem{5577781}
F.~Khan, T.~Ratnarajah, and M.~Sellathurai, ``Multiuser diversity analysis in
  spectrum sharing cognitive radio networks,'' in \emph{Cognitive Radio
  Oriented Wireless Networks Communications (CROWNCOM), 2010 Proceedings of the
  Fifth International Conference on}, June 2010, pp. 1 --5.

\bibitem{4100173}
A.~Ghasemi and E.~S. Sousa, ``Fundamental limits of spectrum-sharing in fading
  environments,'' \emph{Wireless Communications, IEEE Transactions on}, vol.~6,
  no.~2, pp. 649 --658, Feb. 2007.

\bibitem{5454289}
A.~Tajer and X.~Wang, ``Multiuser diversity gain in cognitive networks,''
  \emph{Networking, IEEE/ACM Transactions on}, vol.~18, no.~6, pp. 1766 --1779,
  Dec. 2010.

\bibitem{5054705}
------, ``Multiuser diversity gain in cognitive networks with distributed
  spectrum access,'' in \emph{Information Sciences and Systems, 2009. CISS
  2009. 43rd Annual Conference on}, March 2009, pp. 135 --140.

\bibitem{6279522}
Z.~Bouida, K.~Qaraqe, M.~Abdallah, and M.-S. Alouini, ``Performance analysis of
  joint multi-branch switched diversity and adaptive modulation schemes for
  spectrum sharing systems,'' \emph{Communications, IEEE Transactions on},
  vol.~60, no.~12, pp. 3609--3619, December.

\bibitem{la_cam}
L.~L. Cam, ``An approximation theorem for the poisson binomial distribution,''
  \emph{Pacific Journal of Mathematics}, vol.~10, no.~4, pp. 1181 --1197,
  December 1960.

\bibitem{shaked_stochastic_1994}
M.~Shaked and J.~G. Shanthikumar, \emph{Stochastic orders and their
  applications}.\hskip 1em plus 0.5em minus 0.4em\relax Academic Press, 1994.

\bibitem{stochastic_ordering}
------, \emph{Stochastic Orders}.\hskip 1em plus 0.5em minus 0.4em\relax
  Springer, 2006.

\bibitem{mueller_comparison_2002}
A.~Muller and D.~Stoyan, \emph{Comparison methods for stochastic models and
  risks}.\hskip 1em plus 0.5em minus 0.4em\relax John Wiley and Sons, 2002.

\bibitem{feller_introduction_2009}
W.~Feller, \emph{An Introduction To Probability Theory And Its Applications,
  {2Nd} Ed}.\hskip 1em plus 0.5em minus 0.4em\relax Wiley India Pvt. Ltd.,
  2009.

\bibitem{Ineq}
A.~W. Marshall, I.~Olkin, and B.~C. Arnold, \emph{Inequalities: Theory of
  Majorization and Its Applications}, 2nd~ed.\hskip 1em plus 0.5em minus
  0.4em\relax Springer New York Dordrecht Heidelberg London, 2011.

\bibitem{PGFbound}
D.~J. Daley. and P.~Narayan, ``Series expansions of probability generating
  functions and bounds for the extinction probability of a branching process,''
  \emph{Journal of Applied Probability}, vol.~17, no.~4, pp. 939 --947,
  December 1980.

\bibitem{bigO}
N.~de~Brujin, \emph{Asymptotic Methods In Analysis}.\hskip 1em plus 0.5em minus
  0.4em\relax Amsterdam: North-Holland, 1958.

\bibitem{6112148}
A.~B. Narasimhamurthy, C.~Tepedelenlioglu, and Y.~Zhang, ``Multi-user diversity
  with random number of users,'' \emph{Wireless Communications, IEEE
  Transactions on}, vol.~11, no.~1, pp. 60 --64, January 2012.

\bibitem{2}
D.~N.~C. Tse and P.~Viswanath, \emph{Fundamentals of Wireless Communication},
  1st~ed.\hskip 1em plus 0.5em minus 0.4em\relax Cam-bridge: Cambridge
  University Press, Jun. 2005.

\bibitem{Downey93anabelian}
P.~J. Downey, ``An {Abelian} theorem for completely monotone functions,'' 1993,
  {available online at ftp://ftp.cs.arizona.edu/ reports/1993/TR93-15.ps.Z}.

\end{thebibliography}

\end{document}